\documentclass[journal,twoside,web]{ieeecolor}

\usepackage{generic}

\usepackage{cite}
\usepackage{amsmath,amssymb,amsfonts, amsthm}
\usepackage{algorithmic}
\usepackage{graphicx}
\usepackage{textcomp}
\usepackage{xcolor}
\usepackage{todonotes}

\newcommand{\strategy}{\gamma}
\newcommand{\strategySpace}{\mathcal{S}}
\newcommand{\belief}{\mu}
\newcommand{\beliefSpace}{\mathcal{B}}
\newcommand{\type}{\theta}
\newcommand{\typeSpace}{\Theta}
\newcommand{\prior}{\pi}
\newcommand{\feedbackspace}{\mathfrak{L}}
\newcommand{\criticalTime}{\tilde{s}}

\newtheorem{definition}{Definition}
\newtheorem{proposition}{Proposition}
\newtheorem{assumption}{Assumption}
\newtheorem{remark}{Remark}

\def\BibTeX{{\rm B\kern-.05em{\sc i\kern-.025em b}\kern-.08em
    T\kern-.1667em\lower.7ex\hbox{E}\kern-.125emX}}

\title{On Type Deception in Linear-Quadratic Differential Games}

\author{Jesse Milzman and Dipankar Maity
\vspace{-5mm}
\thanks{\copyright~2026 IEEE. Personal use of this material is permitted. Permission from IEEE must be obtained for all other uses, in any current or future media, including reprinting/republishing this material for advertising or promotional purposes, creating new collective works, for resale or redistribution to servers or lists, or reuse of any copyrighted component of this work in other works.}
\thanks{This research is supported by the ARL grant ARL DCIST CRA W911NF-17-2-0181. 
The views expressed in this paper are those of the authors and do not reflect the official policy or position of the United States Government, Department of War,
or its components.}
\thanks{J.~Milzman is with the DEVCOM Army Research Laboratory, Brooklyn, New York USA.
} 
\thanks{D.~Maity is with the Department of Electrical and Computer Engineering, University of North Carolina at Charlotte, NC, 28223, USA. 
{\tt jesse.m.milzman.civ@army.mil, dmaity@charlotte.edu}
}%
}



\begin{document}

\maketitle
\thispagestyle{empty}
\pagestyle{empty}

\begin{abstract}
We consider two-player linear-quadratic differential games of incomplete information, in which one player has a private type initially unknown to the other. The typed player has incentive to conceal their type, while the uninformed player has the potential to infer it during play.
Any ex-ante equilibrium in this setting will decompose into a deceptive, pooling phase, and a complete-information, revelatory phase. We demonstrate how to solve both phases via nested Riccati equations.
Candidate equilibria are then found by maximizing the game value over a scalar revelation time, for which we provide a gradient in the case of time-homogeneous system matrices.
We conclude by demonstrating our framework in a pursuit-evasion game with time-varying control advantages, finding interior optimal revelation times that confirm deception has quantifiable ex-ante value.
\end{abstract}



\section{Introduction}
\label{sec:introduction}

Classical differential game theory implicitly assumes that all players share a \textit{common understanding} of the system dynamics and objective functions. Under this assumption, equilibrium strategies can be characterized through the Hamilton-Jacobi-Isaacs (HJI) equation \cite{isaacs1999differential}, or through coupled Riccati equations in the special case of linear-quadratic (LQ) games \cite{bacsar2008h}.
The problem becomes even more challenging when one player recognizes that its model assumptions are less accurate than those of its opponent, and the opponent is aware of this  asymmetry \cite{shishika2025deception}. 
In such settings, the player with superior knowledge may attempt to strategically exploit this information advantage, while the other player must adopt cautious strategies to mitigate potential exploitation. 
These situations arise naturally in adversarial settings where agents possess different levels of situational awareness or model fidelity.

In this paper, we study differential games with \textcolor{black}{asymmetric information}, where agents possess different levels of \textcolor{black}{knowledge} regarding the system model and are aware of this asymmetry. Relaxing the common knowledge assumption leads to a fundamentally different strategic structure for the resulting game. In particular, we show that the overall game can be decomposed into two sequential phases. 
In the first phase, the player with superior knowledge deliberately conceals their advantage by behaving in a way that prevents the opponent from inferring its type.
In the second phase, the player reveals their type, after which both players operate under the same knowledge.
The optimal revelation time (equivalently, the duration of the first phase relative to the total game duration) becomes a key strategic decision for the informed player.
Such deliberate concealment to preserve informational advantage builds on an underexplored facet of differential games.

Games with \textit{incomplete information} have been studied extensively in the literature, with the work of \cite{harsanyi1967games} among the earliest contributions, establishing Bayesian games in the static setting.
In the dynamic setting,
\cite{cardaliaguet2012games} considers a class of dynamic games in which the payoff function type is unknown to one or both players.
The authors in \cite{cardaliaguet2014pure} consider \textcolor{black}{instead} a differential game variant in which the information asymmetry arises because one agent knows the initial state, whereas the other agent only knows the distribution from which the initial state is drawn.
This line of work closely relates to studies on measurement asymmetry, where players observe noisy or partial measurements of the system state \cite{bagchi1981linear, maity2017linear, maity2017stochastic}.
\textcolor{black}{None of these works, however, address themselves to the possibility of deception or concealment in dynamic or differential games.}

\textcolor{black}{There is some work concerned with deceptive belief space dynamics, most notably \cite{huang2021dynamic}.
However, deception here emerges  because it is induced by hand-crafted cost functions that reward trajectory ambiguity, instead of emerging naturally as a means of information advantage.}
Our work is most closely related to the setting considered in \cite{shishika2025deception}, where one agent knows the true system dynamics while the other agent only knows a set of possible dynamics. 
Similar to ours, the more informed player has two types (fast and slow), who align their dynamics to conceal their type for a period, before the fast player reveals their type to exploit their advantage.
This previous work, however, is focused on the study of a very particular differential game, and does not provide a more systematic understanding of the underlying phenomenon of deception.

\textbf{Contribution.} To the best of our knowledge, we contribute here the first investigation of capability deception within two-player, zero-sum LQ games with asymmetric type information.
In particular, we specify the ex-ante equilibrium concept that emerges from careful study of the players' information structures.
We prove that any ex-ante equilibrium decomposes into a deceptive ambiguity phase followed by complete-information play. We show that this ambiguous phase reduces to a tractable LQ formulation, parametrized by revelation time. For constant dynamics, we provide an analytic gradient of the game value with respect to this revelation time. Finally, we apply our approach to a time-varying pursuit-evasion game, demonstrating that deception has quantifiable ex-ante value.


\section{Preliminaries}
\label{sec:preliminaries}

We review a two-player, zero-sum LQ (ZSLQ) game of finite horizon $T > 0$.
Let $x(t) \in \mathbb{R}^{n}$ be the system state and $u^i(t) \in \mathbb{R}^{m^i}$ be player P$i$'s control input, for $i=1,2$. The system dynamics and cost functional are given as
\begin{gather}
    \label{eq:zero_sum_LQ.dynamics}
    \dot{x} = A x + B^1 u^{1} + B^{2} u^{2} \\
    \label{eq:zero_sum_LQ.cost}
    J = \int_{0}^T  \left\| x \right\|_{Q}^2 + \left\| u^1 \right\|_{R^1}^2 -  \left\| u^{2} \right\|_{R^2}^2 \, dt + \left\| x(T) \right\|_{Q_f}^2
\end{gather}
where P1 minimizes and P2 maximizes $J$; $Q, Q_f \geq 0$, $R^i > 0$,
and all system matrices are assumed at least piecewise continuous in time, though we typically omit $t$ as an argument.
It is well known that if the Riccati equation~\eqref{eq:LQ riccati} admits a solution over $[0,T]$, then the unique~linear feedback Nash equilibrium is given by its solution~\cite{lukes1971equilibrium,bacsar1998dynamic}:
\begin{gather}
    \label{eq:LQ riccati}
    \dot{P} + A^\top P + P A + Q - PSP = 0, \quad P(T) = Q_f, \\
    S = B^1 (R^1)^{-1} B^{1\top} - B^2 (R^2)^{-1} B^{2\top}.
\end{gather}
In this case, the equilibrium feedback strategies and game value are given by
\begin{gather} 
u^i = (-1)^{i} B^{i\top} (R^i)^{-1} P x, \\
J = \|x(0)\|_{P(0)}^2.
\end{gather}
One can typically take both $R^i=I$ without loss of generality, and in much of our paper we consider such systems.


\section{Problem Formulation}
\label{sec:problem_formulation}

\subsection{ZSLQ Game with Asymmetric Type Uncertainty}

Consider again the game from \eqref{eq:zero_sum_LQ.dynamics}-\eqref{eq:zero_sum_LQ.cost}, and take $R_i=I$.
Now, however, assume that P2's control matrix $B^2$ depends on an exogenous random variable $\type$ taking values in $\Theta = \{\type_1, \type_2 \}$.
We assume a common prior $\pi \in \Delta(\Theta)$ over this variable.
Formally, for a given $\type_k$ and open-loop controls $\{u^i(t)\}_{t \in [0,T], i=1,2}$ the realized dynamics~yield:
\begin{gather}
    \label{eq:zero_sum_LQ.dynamics.P2_types}
    \dot{x} = A x + B^1 u^{1} + B^{2}(\type) u^{2} .
\end{gather}
We assume that P2 receives their type $\type_k$ prior to play --- however, P1 does not. From their observations, P1 is able to infer the total input $B^2 u^2$, but not $u^2$ itself, and thus cannot infer $\type$ except insofar as P2's type-dependent strategies result in different trajectories.

To make this precise, we exclusively consider systems with a \emph{private state decomposition},\footnote{This simplifies the exposition and the discussion of the admissible strategy space.} in which the full state decomposes as
$x = (x^1, x^2)$ with
\begin{gather}
    B^1 = \begin{bmatrix} \tilde{B}^1 \\ 0 \end{bmatrix}, \quad
    B^{2}(\type) = \begin{bmatrix} 0 \\ \tilde{B}^{2}(\type) \end{bmatrix}.
\end{gather}
In this setting, each player's control affects only their own state component $x^i$. \textcolor{black}{These components naturally emerge in systems where each player has their own dynamics, e.g. are each physically embodied with their own pose and velocity.} P1 observes the full state $(x(t), \dot{x}(t))$, and thus can reconstruct
$\tilde{B}^{2}(\type) u^2$ from $\dot{x^2}$, but cannot distinguish the control input from its type-specific modifier.
We do, however, impose the following assumption on the two $\type$-specific complete-information subgames defined by (\ref{eq:zero_sum_LQ.cost}) and~(\ref{eq:zero_sum_LQ.dynamics.P2_types}):
\begin{assumption}
    \label{asm:riccati_distinguishability}
    For each $\type_k \in \Theta$, let $P_k$ be the solution to the Riccati equation (\ref{eq:LQ riccati}) under dynamics (\ref{eq:zero_sum_LQ.dynamics.P2_types}). For each $t \in [0,T]$, we assume that 
    \begin{equation}
        B^2(\type_1) B^2(\type_1)^\top P_1(t) x \neq B^2(\type_2) B^2(\type_2)^\top P_2(t) x
    \end{equation}
    for any $x \neq 0$.
\end{assumption}
Additionally, it will simplify our later investigation greatly if we may take the following for granted, as well:
\begin{assumption}
    \label{asm:control_invertability}
    $\tilde{B}^2(\type)$ is invertible for all $t$ and type $\type \in \typeSpace$.
\end{assumption}

If our game considered only open-loop strategies, in which each player commits to their control trajectory at the outset, we could solve this fairly easily as an infinite-dimensional Bayesian game. However, if P1 can observe the trajectory and update their belief, they could discover $\type$ and adapt, which in turn introduces the incentive for P2 to \emph{conceal} their type.
This tension between P1's capacity to infer $\type$ and P2's incentive to conceal it  is the central object of study in this work.
In order to make this precise, we need to clarify (i)~both players' strategy spaces, which in turn depend on the proper definition of their (ii)~information structures, and in particular (iii)~the bookkeeping that accounts for P1's current information regarding $\type$.
We cover all three of these in the next section.

\subsection{Player Information and Strategy Spaces}

A full treatment of the problem we have sketched would require a Perfect Bayesian Equilibrium (PBE) setup \cite{fudenberg1991game}, the theory of which is largely underdeveloped for differential games.
We instead build the simplest strategy and information structures that capture the key informational tension between the players, leaving the full PBE development to future work.

In the usual two-player ZSLQ setting, each player picks a linear feedback strategy $\strategy^i(x(t),t) = K^i(t) x(t)$ for their control $u^i(t) = \strategy^i(x(t),t)$.
Since P2 receives complete information of the game at time $t=0$, we can reasonably reduce their strategy to the choice of one such linear feedback strategy per type.
Formally, let $\mathfrak{L}^i$ be the set of time-continuous linear feedback matrices for P$i$'s control --- i.e. the collection of matrix-valued continuous functions of the form $K^i: [0,T] \to \mathbb{R}^{n \times m^i}$.
Then P2's strategy will live in the space $\tilde{\mathcal{S}}^2 \triangleq \{ \strategy^2: \Theta \to \mathfrak{L}^2 \}$, i.e., where $\strategy^2(x,t;\type) = K^2_{\type} x$ for some appropriate pair of feedback laws $\{ K^2_\type \}_{\type \in \Theta}$.
Thus, P2's decision-making process proceeds as follows (a) pick some suitable $\strategy^2 \in \tilde{\mathcal{S}}^2$, (b) receive $\type$ at $t=0$, (c) execute feedback strategy $\strategy^2( \cdot, \cdot ; \type)$, implicitly reacting to P1's decisions via state feedback.

Let us now turn our attention to P1. How will they perceive their opponent's strategy, both in anticipation and during their own strategy's execution? As mentioned earlier, we take for granted that they can reconstruct $B^2(\type) u^2$. During execution, let $\mu(\type ; t)$ denote their belief over $\Theta$ at time~$t$. Based on the observability mentioned above, it stands to reason that for a given P2 strategy $\strategy^2$, at time $t \in [0,T]$, P1 might find themselves in one of the two following information~states:
\begin{itemize}
    \item If $B^2(\type_1) \strategy^2(x(s),s;\type_1) = B^2(\type_2) \strategy^2(x(s),s;\type_2)$ for all $s \leq t$, then P1 knows no more about $\type$ than they did initially, i.e. $\mu(\cdot ; t) = \pi(\cdot)$.
    \item If there existed an $s \in [0,t]$ where~$B^2(\type_1) \strategy^2(x(s),s;\type_1) $ $ \neq B^2(\type_2) \strategy^2(x(s),s;\type_2)$, then P1 knows $\type = \type^\star$ (the realized type), i.e. $\mu(\type^{\star} ; t) = 1$.
\end{itemize}
This assumes that P1 possesses P2's solution at equilibrium, and is allowed to shape its belief rule based on it.
The former is well-understood in the epistemics of what it means to be at Nash \cite{aumann1995epistemic}, and the latter is actually a requirement of belief functions in the PBE literature: on the equilibrium path, belief functions must be consistent with Bayes' rule as applied to the equilibrium strategies \cite{fudenberg1991game}.

Of course, the astute reader might identify that our description above is too clean. It holds on the equilibrium path, itself, but the equilibrium is defined with respect to P2 defections.
If P2 deviates, P1 may observe a control input that is inconsistent with both equilibrium types, and it is unclear how P1 should update their belief in response. This is precisely why the PBE solution concept requires belief specification on non-equilibrium paths, which is challenging in our differential setting.\footnote{Even if nominally we set the belief on all paths, say based on the relative distance to the two equilibrium paths, how we would go about demonstrating the non-exploitability of such a P1 belief on P2's part is not immediately~obvious.} For this reason, we restrict P2's strategy space to keep the analysis tractable.

Given a better understanding of P1's information structure, we may now define their strategy space. P1's strategy, properly considered, is closed-loop instead of pure feedback, since the belief is a function of the entire history. However, the belief itself is sufficient for the entire previous history, and solely allows P1 to select between the appropriate state-feedback controls. That is, P1 will be playing a strategy of the form:
\begin{equation}
    u^1(t) = \strategy^1(x(t), t; \belief) = K^1_{\belief(t)} \, x(t)
\end{equation}
where $\belief(t) \in \{\prior, \delta_{\type_1}, \delta_{\type_2}\} \triangleq \beliefSpace$ is determined by the observed history $\{B^2(\type) u^2(s)\}_{s \in [0,t]}$ as described above.
Formally, this strategy space is given by $\tilde{\strategySpace}^1 = \{ \strategy^1: \beliefSpace \to \feedbackspace^1 \}$, i.e. a feedback law is chosen for each of the three possible beliefs.

\subsection{The Concealment Interval and Revelation Time}

We have given a semi-formalized account of how P1's belief evolves with the identifiability of P2's type-specific strategies. We now formalize when P1's belief transition actually occurs.

For a choice of P2 strategy $\strategy^2$, we let $K_{\type}^2(t)$ be the associated feedback matrix.
Strictly, it is possible that, for certain states $x$, $B^2(\type_1) K^2_{\type_1}(t) \neq B^2(\type_2) K^2_{\type_2}(t)$ while $B^2(\type_1) K^2_{\type_1}(t) x = B^2(\type_2) K^2_{\type_2}(t) x$, even for $x \neq 0$.
We ignore such degenerate cases, and assume that distinguishability of the matrix operators is equivalent to distinguishability of the resultant controls. With that in mind, for a given strategy $\strategy^2$, we define the \textbf{revelation time} $\criticalTime = \criticalTime(\strategy^2)$ as
\begin{align} 
    \criticalTime &= \min(T, \criticalTime'),\\
    \criticalTime' &= \inf\{ t \in [0,T]~ |~ B^2(\type_1) K^2_{\type_1}(t) \neq B^2(\type_2) K^2_{\type_2}(t) \}.
\end{align}

Observe that, for any $\strategy^2 \in \tilde{\strategySpace}^2$, we have a well-defined $\criticalTime \in [0,T]$.
The idea here is that, for $t \leq \criticalTime$, we have $\mu(\cdot;t)=\pi$, i.e. P1 has not yet identified P2's type.
Moreover, under equilibrium play, as discussed earlier, P1's belief must become certain of the true $\type$ in order for the belief to be consistent.

\subsection{Ex-Ante Nash Equilibrium}

We are almost ready to define the equilibrium concept under investigation: that of players at the root of the game tree, before types have been assigned.
To do so, we need to formalize the notion of belief-consistency within our~setting.

\begin{definition}[Observational compatibility]
\label{defn:observational_compatibility}
A belief rule $\belief$ is \emph{observationally compatible} with a P2 strategy
$\strategy^2 \in \tilde{\strategySpace}^2$ if, for any history $h_t = \{x(s), \dot{x}(s)\}_{s \in [0,t]}$
generated by some $(\strategy^1, x(0), \type)$:
\begin{enumerate}
    \item[(i)] If $B^2(\type_1) K^2_{\type_1}(s) = B^2(\type_2) K^2_{\type_2}(s)$
    for all $s \leq t$, then $\belief(\cdot\,; t) = \prior(\cdot)$.
    \item[(ii)] If there exists $s \leq t$ such that
    $B^2(\type_1) K^2_{\type_1}(s) \neq B^2(\type_2) K^2_{\type_2}(s)$,
    then $\belief(\type^\star; t) = 1$, where $\type^\star$ is the
    realized type.
\end{enumerate}
\end{definition}

The following is an immediate consequence of this and the definition of $\criticalTime$:
\begin{proposition}
\label{prop:obs_compat_implies_on_path}
If $\belief$ is observationally compatible with $\strategy^2$, then
for arbitrary $\strategy^1 \in \tilde{\strategySpace}^1$, $x(0)$, and
realized type $\type^\star \in \typeSpace$, any history generated by
$(\strategy^1, \strategy^2, x(0), \type^\star)$ satisfies:
\begin{enumerate}
    \item[(i)] For $t < \criticalTime(\strategy^2)$,
    $\belief(\cdot\,; t) = \prior(\cdot)$.
    \item[(ii)] For $t \geq \criticalTime(\strategy^2)$,
    $\belief(\type^\star; t) = 1$.
\end{enumerate}
\end{proposition}

\begin{remark}
    Observational compatibility can be thought of as our case-specific variant of the on-path consistency condition that the PBE literature requires of beliefs, at least for equilibrium strategies.
\end{remark}

\begin{remark}
    Note that, if a belief is observationally compatible with strategy $\strategy^2$, then if some $\strategy^{2 \prime}$ agrees with $\strategy^2$ up to some $t>\criticalTime(\strategy^2)$, then the belief is observationally compatible with $\strategy^{2\prime}$ as well.
\end{remark}

We may now define our primary solution concept. For a given strategy profile {\small$(\strategy^1, \strategy^2)$} and realized type $\type$, let {\small$J(\strategy^1, \strategy^2; \type)$} denote the cost (\ref{eq:zero_sum_LQ.cost}) under the dynamics (\ref{eq:zero_sum_LQ.dynamics.P2_types}), and let {\small $\mathbb{E}_\type[J] = \sum_k \prior_k \, J(\strategy^1, \strategy^2; \type_k)$} denote the ex-ante expected cost.

\begin{definition}[Ex-Ante Nash Equilibrium]
\label{def:exAnteNash}
Consider a strategy profile $(\strategy^1, \strategy^2) \in \tilde{\strategySpace}^1 \times \tilde{\strategySpace}^2$
for which there exists a belief rule $\belief$ that is observationally compatible
with $\strategy^2$. We say that this profile is an \emph{ex-ante Nash equilibrium}
if, for all alternative strategies $\hat{\strategy}^1 \in \tilde{\strategySpace}^1$
and $\hat{\strategy}^2 \in \tilde{\strategySpace}^2$,
\begin{subequations}
\label{eqs:exAnteNash}
\begin{align}
    \label{eq:ex_ante_nash.p1}
    \mathbb{E}_\type\left[ J(\hat{\strategy}^1, \strategy^2; \belief) \right]
    &\geq
    \mathbb{E}_\type\left[ J(\strategy^1, \strategy^2; \belief) \right], \\
    \label{eq:ex_ante_nash.p2}
    \mathbb{E}_\type\left[ J(\strategy^1, \hat{\strategy}^2; \belief) \right]
    &\leq
    \mathbb{E}_\type\left[ J(\strategy^1, \strategy^2; \belief) \right].
\end{align}
\end{subequations}
\end{definition}
A consequence of this framework is that we may restrict our attention to better understood subsets of the strategy space.

\begin{proposition}
\label{prop:equilibrium_structure}
Let $(\strategy^1, \strategy^2)$ be an ex-ante Nash equilibrium
with observationally compatible belief $\belief$, and let
$\criticalTime = \criticalTime(\strategy^2)$. Then:
\begin{enumerate}
    \item[(i)] For $t < \criticalTime$, both types of P2 play a
    common feedback law, i.e.
    $B^2(\type_1) K^2_{\type_1}(t) = B^2(\type_2) K^2_{\type_2}(t)$.
    \item[(ii)] For $t \geq \criticalTime$, the strategies coincide
    with the unique complete-information linear feedback equilibrium:
    for each realized type $\type_k$,
    \begin{align}
        K^2_{\type_k}(t) &= B^{2}(\type_k)^\top P_k(t), \\
        K^1_{\delta_{\type_k}}(t) &= -B^{1\top} P_k(t),
    \end{align}
    where $P_k$ solves the Riccati equation
    \emph{(\ref{eq:LQ riccati})} with $B^2 = B^2(\type_k)$
    on $[\criticalTime, T]$.
\end{enumerate}
We allow $\strategySpace^2 \subset \tilde{\strategySpace}^2$ to denote the strategy subspace satisfying (i) and (ii), and $\strategySpace^1 \subset \tilde{\strategySpace}^1$ satisfying (ii).
\end{proposition}
This dramatically simplifies the search space for potential ex-ante equilibria. We defer the proof to Section~\ref{subsec:approach.simplify_strategies}.

However, once P2 learns $\type$, they may prefer to deviate from the pooling arrangement.
Analysis of \textcolor{black}{ex-interim} type-rationality for prospective ex-ante equilibria extends naturally from our framework, though we defer it to future work.


\section{Approach}
\label{sec:approach}

\subsection{Simplification of Strategy Spaces}
\label{subsec:approach.simplify_strategies}

Our entire analysis depends on the equilibrium simplification provided by Proposition~\ref{prop:equilibrium_structure}, and so we begin with its proof.

\begin{proof}[Proof of Proposition~\ref{prop:equilibrium_structure}]
Let our equilibrium $(\strategy^1,\strategy^2)$ be given, satisfying the conditions of Prop.~\ref{prop:equilibrium_structure}.
Part (i) is immediate from the definition of $\criticalTime$.
For part (ii), let $\type_k$ be given and pick $\epsilon > 0$ arbitrary.
For $t \geq \criticalTime + \epsilon$, the belief satisfies $\belief(\type_k; t) = 1$ by Prop.~\ref{prop:obs_compat_implies_on_path}, on the history $h_t$ generated when $\type_k$ is realized.
Consider the restriction of the equilibrium feedback laws $(K^1_{\delta_{\type_k}}(t),\, K^2_{\type_k}(t))$ to $[\criticalTime + \epsilon,\, T]$ --- that is, the feedback laws employed on the realized path where $\type = \type_k$ and $\belief = \delta_{\type_k}$.

We claim these must constitute a Nash equilibrium for the complete-information LQ game on $[\criticalTime + \epsilon, T]$ with $B^2 = B^2(\type_k)$.
Suppose not: then some player $i$ has an improving deviation $K^{i\prime}(t)$ on this interval.
Construct a full ex-ante strategy agreeing with $\strategy^i$ on $[0, \criticalTime + \epsilon]$ and replacing the $\type_k$-path feedback law with $K^{i\prime}$ for $t > \criticalTime + \epsilon$, while leaving all other belief- or type-contingent feedback laws unchanged.
Since the cost on $[0, \criticalTime + \epsilon]$ is unchanged and the belief is certain on this path for $t \geq \criticalTime + \epsilon$, the improvement on the $\type_k$-path carries through to $\mathbb{E}_\theta[J]$ (weighted by $\pi_k > 0$), contradicting \eqref{eq:ex_ante_nash.p1} or \eqref{eq:ex_ante_nash.p2}.
By uniqueness of the continuous linear feedback Nash equilibrium for complete-information LQ games \cite{lukes1971equilibrium,bacsar1998dynamic}, we conclude $K^2_{\type_k}(t) = B^2(\type_k)^\top P_k(t)$ and $K^1_{\delta_{\type_k}}(t) = -B^{1\top} P_k(t)$.
Since $\epsilon > 0$ was arbitrary, this holds on $[\criticalTime, T]$ by continuity of the feedback laws.
\end{proof}

The logic here enables us to break apart all prospective equilibria strategies into two parts: the \textit{asymmetric information regime}, in which both P2 types are constrained to `collaborate,' and then the \textit{complete information regime}, in which they give the game away, so-to-speak.

\subsection{Game Decomposition}

Given the strategy restrictions from Proposition~\ref{prop:equilibrium_structure}, we can see that the cost functionals naturally decompose into two stages of the dynamic portion of the game. For shorthand, let $g(x,u,t) = \left\| x \right\|^2_{Q(t)} + \left\| u^1 \right\|^2 - \left\| u^2 \right\|^2$, and we will also abbreviate $x_{\theta} \triangleq x(t;\theta),$ similarly $u_{\theta}=(u^1_{\theta}, u^2_{\theta})$, as to be defined below.
For any $(\strategy^1,\strategy^2) \in \strategySpace^1 \times \strategySpace^2$, we have that
\begin{align}
    \label{eq:approach.pooled_cost}
    \bar{J}(\strategy^1, \strategy^2) &\triangleq \mathbb{E}_{\theta} J(\strategy^1, \strategy^2;\theta) \\
    \nonumber J(\strategy^1, \strategy^2;\type_k) & = \int_0^{\criticalTime} g(x_{\type_k}, u_{\type_k}, t) \, dt \\
    \label{eq:approach.stage_decomposed_cost}
    &\quad + \int_{\criticalTime}^{T} g(x_{\type_k}, u_{\type_k}, t) \, dt + \left\| x_{\type_k}(T) \right\|^2_{Q_f}
   \end{align}
   where $x_{\type}(t)$ is generated from $u_{\type}$ according to (\ref{eq:zero_sum_LQ.dynamics.P2_types}), and 
   \begin{align} 
    u^i_{\type_k} = 
    \begin{cases}
        K^1_\prior(t) \, x(t), & i=1, t \in [0,\criticalTime), \\
       K^2_{\type_k}(t) \, x(t), & i=2, t \in [0,\criticalTime), \\
        (-1)^i B^i(\type)^\top P_k(t) \, x(t), & t \in [\criticalTime,T],
    \end{cases} \\[6pt]
    \label{eq:decomposition.observability_constraint}
    B^2(\type_1) K^2_{\type_1}(t) = B^2(\type_2) K^2_{\type_2}(t), \quad t \in [0,\criticalTime).
\end{align}
Equation (\ref{eq:decomposition.observability_constraint}), as a constraint, directly emerges from the assumption of observational compatibility of $\belief$ and the definition of $\criticalTime$. This will allow for two great simplifications in our decision space.

First, all P1 and P2 strategies in the restricted space are identical for $t>\criticalTime(\strategy^2)$: they're playing the complete information game past that (strategy-specific) time-point.
Thus, we can reduce that portion of the game to the value of the complete information game:
\begin{equation}
    \label{eq:stage2_cost}
    \left\| x_{\type_k}(\criticalTime) \right\|_{P_k(\criticalTime)}^2 = \int_{\criticalTime}^{T} g(x_{\type_k}, u_{\type_k}, t) \, dt + \left\| x_{\type_k}(T) \right\|^2_{Q_f}.
\end{equation}
Thus, once $\criticalTime$ becomes fixed, the only meaningful choice left to P1 or P2 is how they will behave over $t \in [0,\criticalTime]$. We refer to this pre-$\criticalTime$ portion of the game as \textbf{Stage 1}, i.e., the first integral in the RHS of  (\ref{eq:approach.stage_decomposed_cost}), and the latter, complete information regime given by (\ref{eq:stage2_cost}) as \textbf{Stage 2}.\footnote{Note that these stages only describe on-path behavior.}

This leads us to the second major simplification we can make. Recall that $B^2(\type_k) = [0; \tilde{B}^2(\type_k)]$ and $\tilde{B}^2(\type_k)$ is invertible for each $\type_k$.
We may rewrite (\ref{eq:decomposition.observability_constraint}) only with the lower blocks, i.e. $\tilde{B}^2(\type_1) K^2_{\type_1}(t) = \tilde{B}^2(\type_2) K^2_{\type_2}(t)$, and then proceed to multiply through by $(\tilde{B}^2(\type_2))^{-1}$, giving us
\begin{equation}
    K^2_{\type_2}(t) = \left( \tilde{B}^2(\type_2) \right)^{-1} \tilde{B}^2(\type_1) K^2_{\type_1}(t).
\end{equation}
This tells us, then, that P2's choice of $K^2_{\type_2}(t)$ is entirely determined by their choice of $K^2_{\type_1}(t)$ for $t\in[0,\criticalTime)$. Pulling these two simplifications together, we may reduce each player's relevant decision space further.

\begin{proposition}
\label{prop:reduced_decisions}
Let $(\strategy^1, \strategy^2) \in \strategySpace^1 \times \strategySpace^2$.
Then the strategy pair is fully characterized by:
\begin{enumerate}
    \item[(i)] P2's choice of revelation time $\criticalTime = \criticalTime(\strategy^2)$ and feedback law $K^2_{\type_1}(t)$ for $t \in [0,\criticalTime)$.
    \item[(ii)] P1's choice of feedback law $K^1_\prior(t)$ for $t \in [0,T]$.
\end{enumerate}
All remaining components of the strategy pair are determined by Proposition~\ref{prop:equilibrium_structure}, except the off-path $\{K^1_{\delta_k}(t)\}_{t \in [0,\criticalTime)}$ (see Remark below).
\end{proposition}
\begin{remark}
    P1's off-path complete information feedback laws $K^1_{\delta_{\theta_k}}(t)$ for $t < \criticalTime$ do not affect the ex-ante equilibrium conditions --- unlike the off-path pooling law $K^1_\prior(t)$ for $t \geq \criticalTime$, which can affect them.
For ex-ante analysis, we may arbitrarily take them as $K^1_{\delta_{\theta_k}} = -B^{1\top} P_k(t)$ for all $t \in [0,T]$.
\end{remark}

\subsection{Phase 1 ZSLQ Game}

Given Proposition~\ref{prop:equilibrium_structure}, we may now approach the search~for ex-ante equilibria as merely solving the Stage 1 game described in the previous section.
As established, $\criticalTime$ may be viewed as a component of P2's decision. Suppose that $\criticalTime$ is fixed, and the two players are simultaneously choosing strategies within $\strategySpace^1 \times \strategySpace^2_{\criticalTime}$, where $\strategySpace^2_{\criticalTime'} = \{\strategy^2 \in \strategySpace^2 ~|~ \criticalTime(\strategy^2) = \criticalTime' \}$.
We can then characterize the resulting equilibrium.

\begin{proposition}
\label{prop:stage1_LQ}
Fix $\criticalTime \in [0,T]$ and consider the ex-ante Nash equilibrium conditions \eqref{eq:ex_ante_nash.p1}--\eqref{eq:ex_ante_nash.p2} restricted to $\strategySpace^1 \times \strategySpace^2_{\criticalTime}$.
Then the ex-ante expected cost $\bar{J}(\criticalTime)$ is equal to the value of a ZSLQ game on $[0,\criticalTime]$ with dynamics and cost
\begin{equation}
    \label{eq:stage1_dynamics}
    \dot{x} = Ax + B^1 u^1 + B^2(\type_1) u^2
\end{equation}
\begin{equation}
    \label{eq:stage1_cost}
    \bar{J}(\criticalTime) = \int_0^{\criticalTime} \left\| x \right\|_Q^2 + \left\| u^1 \right\|^2 - \left\| u^2 \right\|_{\tilde{R}^2}^2 \, dt + \left\| x(\criticalTime) \right\|_{\bar{P}(\criticalTime)}^2
\end{equation}
where $u^1 = K^1_\prior(t) x$, $u^2 = K^2_{\type_1}(t) x$, and
\begin{gather}
    \label{eq:R_tilde}
    \tilde{R}^2 = \pi_1 I + \pi_2 \Lambda^\top \Lambda, \quad \Lambda = (\tilde{B}^2(\type_2))^{-1} \tilde{B}^2(\type_1), \\
    \label{eq:P_bar}
    \bar{P}(\criticalTime) = \sum_k \pi_k P_k(\criticalTime).
\end{gather}
The solution is given by the same form as the Riccati equation~(\ref{eq:LQ riccati}), with $R^1=I$, $R^2 = \tilde{R}^2$, $B^2 = B^2(\type_1)$, $Q_f = \bar{P}(\criticalTime)$, and using $M$ instead:
\small{\begin{gather}
    \label{eq:M_riccati} 
    \dot{M} + A^\top M + MA + Q - M\tilde{S}M = 0, \quad M(\criticalTime) = \bar{P}(\criticalTime), \\
    \tilde{S} = B^1 B^{1\top} - B^2(\type_1)(\tilde{R}^2)^{-1} B^2(\type_1)^\top.
\end{gather}}
The Stage 1 equilibrium feedback strategies are
\begin{gather}
    K^1_\prior(t) = -B^{1\top} M(t), \\
    K^2_{\type_1}(t) = (\tilde{R}^2)^{-1} B^2(\type_1)^\top M(t),
\end{gather}
and the equilibrium value is $\bar{J}(\criticalTime) = \left\| x(0) \right\|_{M(0)}^2$.
\end{proposition}
\begin{proof}
Under the constraint \eqref{eq:decomposition.observability_constraint}, both P2 types produce identical closed-loop dynamics on $[0,\criticalTime]$: the input to the state equation is $B^2(\type_k) K^2_{\type_k} x = B^2(\type_1) K^2_{\type_1} x$ for both $k$.
Thus $x_{\type_1}(t) = x_{\type_2}(t)$ for $t \in [0,\criticalTime]$, and we may drop the type subscript on $x$.
Substituting the Stage 2 value \eqref{eq:stage2_cost} into the expected cost, we have
\begin{equation}
    \bar{J} = \int_0^{\criticalTime} \left\| x \right\|_Q^2 + \left\| u^1 \right\|^2 - \mathbb{E}_\theta\left[\left\| u^2_{\type_k} \right\|^2\right] dt + \left\| x(\criticalTime) \right\|_{\bar{P}(\criticalTime)}^2.
\end{equation}
The terminal cost follows from $\mathbb{E}_\theta[\left\| x(\criticalTime) \right\|^2_{P_k(\criticalTime)}] = \left\| x(\criticalTime) \right\|^2_{\bar{P}(\criticalTime)}$ since the $x(\criticalTime)$ is common over $\theta$.
For the running cost, using $K^2_{\type_2} = \Lambda K^2_{\type_1}$ and writing $u^2 = K^2_{\type_1} x$:
\begin{align}
    \mathbb{E}_\theta\left[\left\| u^2_{\type_k} \right\|^2\right] &= \pi_1 \left\| K^2_{\type_1} x \right\|^2 + \pi_2 \left\| \Lambda K^2_{\type_1} x \right\|^2 \nonumber \\
    &= x^\top (K^2_{\type_1})^\top \left(\pi_1 I + \pi_2 \Lambda^\top \Lambda\right) K^2_{\type_1} x \nonumber \\
    &= \left\| u^2 \right\|_{\tilde{R}^2}^2.
\end{align}
This is a ZSLQ game as in Section~\ref{sec:preliminaries}, with $R^2 = \tilde{R}^2$, and the form of the Riccati variable $M$ and the associated solution strategies follow.
\end{proof}

From Proposition~\ref{prop:reduced_decisions}, we know that ex-ante equilibria are found from P1 choosing $\{ K^1_{\prior} \}$ and P2 choosing $(\criticalTime, \{ K^2_{\theta_1} \}_{t \in [0,\criticalTime]})$ simultaneously. 
However, using Prop.~\ref{prop:stage1_LQ} above, it is simpler to proceed via the following, sequential decomposition: (1) P2 picks $\criticalTime$, then (2) P1/P2 choose $K^1_{\prior}/K^2_{\theta_1}$, respectively, on $t\in[0,\criticalTime]$, resulting in game value $\bar{J}(\criticalTime)$.
Any ex-ante equilibrium will necessarily result in a solution to this sequential decomposition.\footnote{Though not vice versa.} Thus, using Prop.~\ref{prop:stage1_LQ}, we have reduced P2's task to the following problem:
\begin{equation}
    \label{eq:min_J}
    \max_{\criticalTime \in [0,T]} \bar{J}(\criticalTime),
\end{equation}
where we note that {\small$\bar{J}(\criticalTime) = \left\| x(0) \right\|_{M(0;\criticalTime)}^2$}. 
Optimizing only over a scalar $\criticalTime$, (\ref{eq:min_J}) can be solved via 0th-order methods (e.g., grid-search, as in Sec.~\ref{sec:experiments}).
However, derivatives can be computed via variational analysis of the Riccati equations, allowing for the use of gradient-based methods.
These have the additional benefit of allowing the verification of first-order necessary optimality conditions.
As an example, we compute the derivative of $\bar{J}(\criticalTime)$ with respect to $\criticalTime$, for the special case of time-invariant LQ~systems.

\begin{proposition}
\label{prop:sensitivity}
Let $M(t;\criticalTime)$ denote the solution to \eqref{eq:M_riccati} at time $t$,
for constant coefficient matrices.
Then
\begin{equation}
    \label{eq:dJds}
    \frac{d\bar{J}}{d\criticalTime} = x(0)^\top \left[ f(M(0)) + N(0) \right] x(0)
\end{equation}
where $f(M) = A^\top M + MA + Q - M\tilde{S}M$, and $N(t)$ solves the linear terminal value problem
\begin{gather}
    \label{eq:variational_N}
    \dot{N} + A^\top N + NA - N\tilde{S}M - M\tilde{S}N = 0, \\
    \label{eq:variational_N_terminal}
    N(\criticalTime) = -A^\top \bar{P} - \bar{P}A - Q + \sum_k \pi_k P_k S_k P_k
\end{gather}
with $S_k = B^1 B^{1\top} - B^2(\type_k) B^2(\type_k)^\top$ and $\bar{P} = \bar{P}(\criticalTime)$. 
\end{proposition}

\begin{proof}
We apply the chain rule:
\begin{equation}
    \frac{d}{d\criticalTime} M(t;\criticalTime) = \frac{\partial}{\partial \criticalTime} M + \frac{\partial M}{\partial \bar{P}(\criticalTime)} \frac{\partial \bar{P}(\criticalTime)}{\partial \criticalTime}.
\end{equation}
For the first term, since all coefficient matrices are constant, $M(t,\criticalTime,\bar{P})$ depends on $t$ and $\criticalTime$ only through $(\criticalTime - t)$, i.e. increasing terminal time to $\criticalTime+\epsilon$ is equivalent to solving backwards a step further to $t-\epsilon$.
Thus $\partial M / \partial \criticalTime = -\partial M / \partial t = f(M(t))$.

For the second term, we use a variational approach.
Perturbing the terminal condition $\bar{P} \mapsto \bar{P} + \epsilon H$ and writing $M(t) + \epsilon N(t) + O(\epsilon^2)$, substitution into \eqref{eq:M_riccati} and collecting $O(\epsilon)$ terms yields \eqref{eq:variational_N} with $N(\criticalTime) = H$.

It remains to compute $H = \partial \bar{P} / \partial \criticalTime$.
Differentiating $\bar{P}(\criticalTime) = \sum_k \pi_k P_k(\criticalTime)$ and substituting each $\dot{P}_k$ from its Riccati equation:
\begin{equation}
    \frac{\partial \bar{P}}{\partial \criticalTime} = \sum_k \pi_k \dot{P}_k(\criticalTime) = -A^\top \bar{P} - \bar{P}A - Q + \sum_k \pi_k P_k S_k P_k
\end{equation}
which gives the terminal condition \eqref{eq:variational_N_terminal}.
Combining both terms and evaluating at $t=0$ yields \eqref{eq:dJds}.
\end{proof}


\section{Numerical Experiments}
\label{sec:experiments}

We demonstrate our approach in a pursuit-evasion game where P1 has an initial control advantage. It is during this period that type deception proves most useful.
Consider a two-player, zero-sum pursuit-evasion game with state
$x(t) = (x^1(t), x^2(t)) \in \mathbb{R}^4$, each $x^i(t)$ denoting P$i$'s position, evolving according to the single-integrator dynamics:
\begin{equation}
\label{eq:PE.dynamics}
    \dot{x} = B^1(t)\, u^1 + B^2(t,\theta)\, u^2.
\end{equation}
Each player's control matrices are given by
\begin{equation}
\small 
    \label{eq:PE.B}
    B^1(t) = \beta_1(t) \begin{bmatrix} I_2 \\ 0 \end{bmatrix}, \qquad
    B^2(t,\theta_k) = \begin{bmatrix} 0 \\ \beta_2(t)\, \tilde{B}^2(\theta_k) \end{bmatrix},
\end{equation}
where
\begin{gather}
\small
    \tilde{B}^2(\theta_1) = \begin{bmatrix} b & 0 \\ 0 & 1 \end{bmatrix}, \quad
    \tilde{B}^2(\theta_2) = \begin{bmatrix} 1 & 0 \\ 0 & b \end{bmatrix}, \qquad b = 1.5, \\
    \beta_i(t) = \beta_i^{\mathrm{base}} + \frac{\beta_i^{\mathrm{boost}}}{1 + e^{(-1)^{i+1}\alpha(t - t_c)}}, \quad i=1,2 \label{eq:num_beta1} 
\end{gather}
The two P2 types possess equal overall control capability, but differ in their preferred axis of motion. Player 1 does not know this favored axis.
Moreover, the scalar weights $\beta_i(t)$ introduce time inhomogeneity, parametrized by $\beta_1^{\mathrm{base}} = 0.2$, $\beta_1^{\mathrm{boost}} = 0.8$,
$\beta_2^{\mathrm{base}} = 0.15$, $\beta_2^{\mathrm{boost}} = 0.65$, and
$t_c = 5$.
Under this parameterization, the pursuer's control effectiveness
decays from approximately $1.0$ to $0.2$, while the evader's grows from
approximately $0.15$ to $0.8$, with the transition governed by the sharpness
parameter $\alpha$.

The cost functional is given by:
\begin{equation}\label{eq:num_cost}
    J = \kappa \, \| x^1(T) - x^2(T) \|^2,
\end{equation}
with $\kappa = 0.1$, corresponding to $Q = 0$ and
$Q_f = \kappa \bigl[\begin{smallmatrix} I_2 & -I_2 \\ -I_2 & I_2 \end{smallmatrix}\bigr]$.
The prior is uniform, $\pi_1 = \pi_2 = 1/2$, and the horizon is $T = 10$,
with initial condition $x_0 = (0,\, 0,\, 1,\, 10)^\top$.

Figure~\ref{fig:PE.alpha_sweep} displays $\bar{J}(\tilde{s})$ for
$\alpha \in \{1, 2, 3, 5\}$. In all cases, the game value exhibits an
interior maximum $\tilde{s}^* \in (0, T)$, indicating that Player~2
benefits from a period of type deception before revelation. As $\alpha$
increases (i.e., the capability transition sharpens), the optimal revelation
time $\tilde{s}^*$ shifts toward the crossover center $t_c = 5$ and the
concealment gain becomes more pronounced.

\begin{figure}
    \centering
    \includegraphics[width=0.8\linewidth]{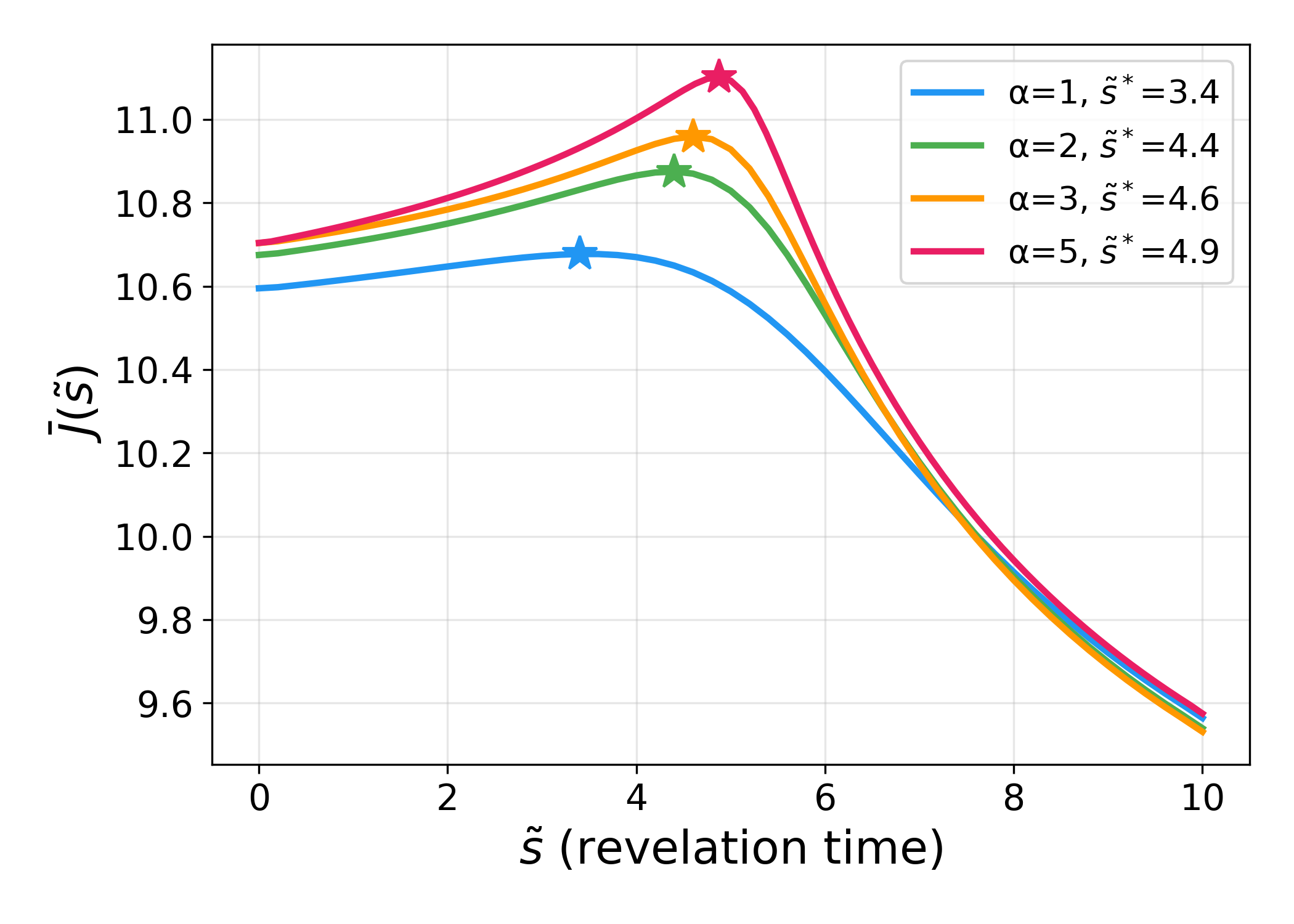}
    \caption{Ex-ante game value $\bar{J}(\tilde{s})$ as a function of revelation time $\tilde{s}$ for varying sigmoid sharpness $\alpha$. Stars indicate optimal revelation times $\tilde{s}^*$, computed via grid-search. All four cases exhibit interior optima, with $\tilde{s}^*$ approaching the crossover center $t_c = 5$ as $\alpha$ increases.}
    \label{fig:PE.alpha_sweep} \vspace{-6 pt}
\end{figure}



\section{Conclusion}
\label{sec:conclusion}

In this work, we developed a framework to introduce asymmetric information to ZSLQ games, in order to enable the possibility of deception in the differential setting.
We provided a principled approach to nominate candidate ex-ante equilibrium strategies, by reducing the game to a two-stage, nested LQ format that can be approached with the familiar Riccati machinery.
Additionally, we demonstrated the plausibility of our approach by applying it to a simple pursuit-evasion game, where ex-ante P2 gains from concealing their type during P1's period of advantage.
Future work includes establishing sufficient conditions for the existence of ex-ante equilibria, and analyzing their ex-interim type-rationality, as well as extending to the case of multiple, privately typed players, in which their respective revelation times themselves may become (a) belief-dependent functions, and (b) comprise a Nash decision in the initial stages.


\bibliographystyle{IEEEtran}
\bibliography{references}

\end{document}